\documentclass[11pt,letterpaper]{article}

\def\showauthornotes{1}

\def\showkeys{0}
\def\showdraftbox{0}
\def\showcolorlinks{1}
\def\usemicrotype{0}
\def\showfixme{0}

\usepackage{etex}

\usepackage[l2tabu, orthodox]{nag}

\usepackage{xspace,enumerate}

\usepackage[dvipsnames]{xcolor}

\usepackage[american]{babel}

\usepackage{mathtools}

\usepackage{amsthm}

\usepackage[ruled,vlined,linesnumbered]{algorithm2e}
\SetKwRepeat{Do}{do}{while}

\newtheorem{theorem}{Theorem}[section]
\newtheorem*{theorem*}{Theorem}

\newtheorem*{proposition*}{Proposition}
\newtheorem{lemma}[theorem]{Lemma}
\newtheorem*{lemma*}{Lemma}

\newtheorem*{conjecture*}{Conjecture}
\newtheorem{fact}[theorem]{Fact}
\newtheorem*{fact*}{Fact}

\newtheorem*{hypothesis*}{Hypothesis}

\theoremstyle{definition}
\newtheorem{definition}[theorem]{Definition}
\newtheorem*{definition*}{Definition}

\newtheorem*{problem*}{Problem}

\theoremstyle{remark}

\newtheorem*{claim*}{Claim}
\newtheorem{remark}[theorem]{Remark}
\newtheorem*{remark*}{Remark}

\newtheorem*{observation*}{Observation}

\usepackage[
letterpaper,
top=1in,
bottom=1in,
left=1in,
right=1in]{geometry}

\usepackage{amsmath,amsfonts,amssymb}

\usepackage{newpxtext} 
\usepackage{textcomp} 
\usepackage[varg,bigdelims]{newpxmath}
\usepackage[scr=rsfso]{mathalfa}
\usepackage{bm} 
\let\mathbb\varmathbb

\ifnum\showkeys=1
\usepackage[color]{showkeys}
\fi

\ifnum\showcolorlinks=1
\usepackage[
colorlinks=true,
urlcolor=blue,
linkcolor=blue,
citecolor=OliveGreen,
]{hyperref}
\fi

\ifnum\showcolorlinks=0
\usepackage[
pagebackref,
colorlinks=false,
pdfborder={0 0 0}
]{hyperref}
\fi

\usepackage{prettyref}
\newcommand{\pref}{\prettyref}

\newcommand{\savehyperref}[2]{\texorpdfstring{\hyperref[#1]{#2}}{#2}}

\newrefformat{eq}{\savehyperref{#1}{\textup{(\ref*{#1})}}}
\newrefformat{prog}{\savehyperref{#1}{\textup{Program \ref*{#1}}}}
\newrefformat{eqn}{\savehyperref{#1}{\textup{(\ref*{#1})}}}
\newrefformat{lem}{\savehyperref{#1}{Lemma~\ref*{#1}}}
\newrefformat{def}{\savehyperref{#1}{Definition~\ref*{#1}}}
\newrefformat{thm}{\savehyperref{#1}{Theorem~\ref*{#1}}}
\newrefformat{cor}{\savehyperref{#1}{Corollary~\ref*{#1}}}
\newrefformat{cha}{\savehyperref{#1}{Chapter~\ref*{#1}}}
\newrefformat{sec}{\savehyperref{#1}{Section~\ref*{#1}}}
\newrefformat{app}{\savehyperref{#1}{Appendix~\ref*{#1}}}
\newrefformat{tab}{\savehyperref{#1}{Table~\ref*{#1}}}
\newrefformat{fig}{\savehyperref{#1}{Figure~\ref*{#1}}}
\newrefformat{hyp}{\savehyperref{#1}{Hypothesis~\ref*{#1}}}
\newrefformat{algo}{\savehyperref{#1}{Algorithm~\ref*{#1}}}
\newrefformat{rem}{\savehyperref{#1}{Remark~\ref*{#1}}}
\newrefformat{item}{\savehyperref{#1}{Item~\ref*{#1}}}
\newrefformat{step}{\savehyperref{#1}{step~\ref*{#1}}}
\newrefformat{conj}{\savehyperref{#1}{Conjecture~\ref*{#1}}}
\newrefformat{fact}{\savehyperref{#1}{Fact~\ref*{#1}}}
\newrefformat{prop}{\savehyperref{#1}{Proposition~\ref*{#1}}}
\newrefformat{prob}{\savehyperref{#1}{Problem~\ref*{#1}}}
\newrefformat{claim}{\savehyperref{#1}{Claim~\ref*{#1}}}
\newrefformat{relax}{\savehyperref{#1}{Relaxation~\ref*{#1}}}
\newrefformat{red}{\savehyperref{#1}{Reduction~\ref*{#1}}}
\newrefformat{part}{\savehyperref{#1}{Part~\ref*{#1}}}
\newrefformat{obs}{\savehyperref{#1}{Observation~\ref*{#1}}}
\newrefformat{corr}{\savehyperref{#1}{Corollary~\ref*{#1}}}

\newcommand{\Sref}[1]{\hyperref[#1]{\S\ref*{#1}}}

\usepackage{nicefrac}

\ifnum\usemicrotype=1
\usepackage{microtype}
\fi

\ifnum\showauthornotes=1
\newcommand{\Authornote}[2]{{\sffamily\small\color{red}{[#1: #2]}}}
\newcommand{\Authornotecolored}[3]{{\sffamily\small\color{#1}{[#2: #3]}}}
\newcommand{\Authorcomment}[2]{{\sffamily\small\color{gray}{[#1: #2]}}}
\newcommand{\Authorstartcomment}[1]{\sffamily\small\color{gray}[#1: }

\newcommand{\Authorfnote}[2]{\footnote{\color{red}{#1: #2}}}
\newcommand{\Authorfixme}[1]{\Authornote{#1}{\textbf{??}}}
\newcommand{\Authormarginmark}[1]{\marginpar{\textcolor{red}{\fbox{\Large #1:!}}}}
\else
\newcommand{\Authornote}[2]{}
\newcommand{\Authornotecolored}[3]{}
\newcommand{\Authorcomment}[2]{}
\newcommand{\Authorstartcomment}[1]{}

\newcommand{\Authorfnote}[2]{}
\newcommand{\Authorfixme}[1]{}
\newcommand{\Authormarginmark}[1]{}
\fi

\ifnum\showfixme=0

\fi

\usepackage{boxedminipage}

\newcommand{\norm}[1]{\lVert#1\rVert}
\newcommand{\Norm}[1]{\left\lVert#1\right\rVert}

\newcommand{\iprod}[1]{\langle#1\rangle}

\newcommand{\Esymb}{\mathbb{E}}
\newcommand{\Psymb}{\mathbb{P}}

\DeclareMathOperator*{\E}{\Esymb}

\DeclareMathOperator*{\ProbOp}{\Psymb}
\DeclareMathOperator*{\pE}{{\tilde\Esymb}}

\renewcommand{\Pr}{\ProbOp}

\DeclareMathOperator*{\argmin}{argmin}

\newcommand{\textparen}[1]{\text{(#1)}}

\ifx\because\undefined
\newcommand{\because}[1]{\textparen{because #1}}
\else
\renewcommand{\because}[1]{\textparen{because #1}}
\fi

\newcommand{\defeq}{\stackrel{\mathrm{def}}=}

\newcommand{\seteq}{\mathrel{\mathop:}=}

\newcommand{\mper}{\,.}
\newcommand{\mcom}{\,,}

\newcommand\bdot\bullet

\ifx\mathds\undefined 

\else

\fi

\DeclareMathOperator{\Span}{Span}

\DeclareMathOperator{\Id}{\mathrm{Id}}

\DeclareMathOperator{\Tr}{Tr}

\newcommand{\etal}{et al.\xspace}

\newcommand{\Z}{\mathbb Z}
\newcommand{\N}{\mathbb N}
\newcommand{\R}{\mathbb R}

\newcommand{\cA}{\mathcal A}
\newcommand{\cB}{\mathcal B}

\newcommand{\cD}{\mathcal D}

\newcommand{\cN}{\mathcal N}

\newcommand{\cP}{\mathcal P}
\newcommand{\cQ}{\mathcal Q}

\newcommand{\cS}{\mathcal S}

\newcommand{\cV}{\mathcal V}

\usepackage{mathrsfs}

\ifnum\showdraftbox=1

\else

\fi

\let\epsilon=\varepsilon

\numberwithin{equation}{section}

\newcommand\MYcurrentlabel{xxx}

\newcommand{\MYstore}[2]{%
  \global\expandafter \def \csname MYMEMORY #1 \endcsname{#2}%
}

\newcommand{\MYload}[1]{%
  \csname MYMEMORY #1 \endcsname%
}

\newcommand{\MYnewlabel}[1]{%
  \renewcommand\MYcurrentlabel{#1}%
  \MYoldlabel{#1}%
}

\newcommand{\MYdummylabel}[1]{}

\newcommand{\torestate}[1]{%
  \let\MYoldlabel\label%
  \let\label\MYnewlabel%
  #1%
  \MYstore{\MYcurrentlabel}{#1}%
  \let\label\MYoldlabel%
}

\newcommand{\restatetheorem}[1]{%
  \let\MYoldlabel\label
  \let\label\MYdummylabel
  \begin{theorem*}[Restatement of \prettyref{#1}]
    \MYload{#1}
  \end{theorem*}
  \let\label\MYoldlabel
}

\newcommand{\restatedef}[1]{%
  \let\MYoldlabel\label
  \let\label\MYdummylabel
  \begin{definition*}[Restatement of \prettyref{#1}]
    \MYload{#1}
  \end{definition*}
  \let\label\MYoldlabel
}

\newcommand{\restatelemma}[1]{%
  \let\MYoldlabel\label
  \let\label\MYdummylabel
  \begin{lemma*}[Restatement of \prettyref{#1}]
    \MYload{#1}
  \end{lemma*}
  \let\label\MYoldlabel
}

\newcommand{\restateprop}[1]{%
  \let\MYoldlabel\label
  \let\label\MYdummylabel
  \begin{proposition*}[Restatement of \prettyref{#1}]
    \MYload{#1}
  \end{proposition*}
  \let\label\MYoldlabel
}

\newcommand{\restatefact}[1]{%
  \let\MYoldlabel\label
  \let\label\MYdummylabel
  \begin{fact*}[Restatement of \prettyref{#1}]
    \MYload{#1}
  \end{fact*}
  \let\label\MYoldlabel
}

\newcommand{\restateobs}[1]{%
  \let\MYoldlabel\label
  \let\label\MYdummylabel
  \begin{observation*}[Restatement of \prettyref{#1}]
    \MYload{#1}
  \end{observation*}
  \let\label\MYoldlabel
}
\newcommand{\restate}[1]{%
  \let\MYoldlabel\label
  \let\label\MYdummylabel
  \MYload{#1}
  \let\label\MYoldlabel
}

\let\origparagraph\paragraph
\renewcommand{\paragraph}[1]{\origparagraph{#1.}}

\allowdisplaybreaks

\usepackage[newenum,newitem,increaseonly]{paralist}
\usepackage{enumitem}

\usepackage{comment}

\let\citet\cite

\theoremstyle{definition}

\DeclareUrlCommand\email{}

\newcommand{\restateproblem}[2]{%
  \let\MYoldlabel\label
  \let\label\MYdummylabel
  \begin{problem*}[Restatement of \prettyref{#1}, {#2}]
    \MYload{#1}
    \end{problem*}
  \let\label\MYoldlabel
}
\usepackage{bm}

\usepackage{tikz,bbm}

\newcommand\whichfont{1}
\ifnum\whichfont=1

\newcommand{\sumn}{\frac{1}{M}\sum\limits_{i=1}^N}

\else
\fi

\newcommand{\mE}{\mathop{\mathbb{E}}}
\usepackage{turnstile}
\newcommand{\SoSp}[1]{\sststile{#1}{}}

\usepackage{thmtools}
\usepackage{thm-restate}

\usepackage{hyperref}

\usepackage{cleveref}

\newcommand{\calA}{\mathcal A}

\newcommand{\calG}{\mathcal G}
\newcommand{\calH}{\mathcal H}

\newcommand{\calQ}{\mathcal Q}
\newcommand{\calR}{\mathcal R}

\newcommand{\calX}{\mathcal X}
\newcommand{\calY}{\mathcal Y}
\newcommand{\calZ}{\mathcal Z}

\newcommand{\bw}{\boldsymbol{w}}

\newcommand{\bcalA}{\boldsymbol{\calA}} 
\newcommand{\bcalG}{\boldsymbol{\calG}} 
\newcommand{\bcalH}{\boldsymbol{\calH}} 
\newcommand{\bcalQ}{\boldsymbol{\calQ}} 
\newcommand{\bcalR}{\boldsymbol{\calR}} 
\newcommand{\bcalX}{\boldsymbol{\calX}} 
\newcommand{\bcalY}{\boldsymbol{\calY}} 
\newcommand{\bcalZ}{\boldsymbol{\calZ}}

\usepackage{tikz}

\usepackage{relsize}
\usepackage[font=footnotesize]{caption}
\usepackage[flushleft,para]{threeparttable}
\makeatletter
\g@addto@macro\TPT@defaults{\footnotesize}
\makeatother

\DeclareUrlCommand\email{}

\renewcommand{\bf}{\textbf}
\renewcommand{\it}{\em}

\DeclareUrlCommand\email{}

\let\pref=\prettyref

\begin{document}

\title{List Decodable Subspace Recovery}
\author{Prasad Raghavendra\thanks{University of California, Berkeley, research supported by NSF Grant CCF 1718695.} \and Morris Yau \thanks{University of California, Berkeley, research supported by NSF Grant CCF 1718695.}}
\maketitle
\thispagestyle{empty}

\begin{abstract}
Learning from data in the presence of outliers is a fundamental problem in statistics.  In this work, we study robust statistics in the presence of overwhelming outliers for the fundamental problem of subspace recovery.  Given a dataset where an $\alpha$ fraction (less than half) of the data is distributed uniformly in an unknown $k$ dimensional subspace in $d$ dimensions, and with no additional assumptions on the remaining data, the goal is to recover a succinct list of $O(\frac{1}{\alpha})$ subspaces one of which is nontrivially correlated with the planted subspace.  We provide the first polynomial time algorithm for the 'list decodable subspace recovery' problem, and subsume it under a more general framework of list decoding over distributions that are "certifiably resilient" capturing state of the art results for list decodable mean estimation and regression.           
\end{abstract}

\clearpage

\tableofcontents
\clearpage

\section{Introduction}

A large hurdle for the deployment of algorithms in high dimensional statistics is their susceptibility to outliers.  The central paradigm of statistical inference is finding the parameters of a statistical model given data.  A long line of work in the robust statistics literature, models 'real world' data as a distributional perturbation of a parameterized generative model $\cD$.  Here robust estimators have been designed for decades, see \cite{huber2011robust}.  Under the classic "Huber Contamination Model" data $X_1,...,X_N$ is drawn i.i.d from a distribution that is a mixture of an inlier distribution $\cD$ belonging to a parameterized family, and an outlier distribution $A$ which can be chosen adversarially.  
$$X_1,...,X_N \sim \alpha D + (1 - \alpha)A $$
Here $\alpha$ is a constant in $[0,1]$ corresponding to the fraction of the dataset that is comprised of inliers and is presumed to be known.  The goal is to recover the relevant parameters of $\cD$ such as mean, covariance, etc.  For $\alpha > 1/2$ the inliers overwhelm the outliers, and we are in the setting of classical robust statistics for which a recent flurry of computationally tractable algorithms have been developed, see survey \cite{dk19survey}.

Less well understood are the settings in which high dimensional statistical inference is possible in the presence of overwhelming outliers.  For $\alpha < 1/2$, we are in the setting of overwhelming outliers, where returning a single estimator for relevant parameters of $\cD$ is impossible as the outlier distribution $A$ can belong to the same distributional family as $D$ but with wildly different parameters.  With the outliers outnumbering the inliers, there is no unique identification of parameters, a problem we refer to as a "failure of identifiability".  However, one could hope to output a short list of estimators of length $O(\frac{1}{\alpha})$, one of which is guaranteed to be close to the true parameters of $\cD$.  \cite{CharikarSV17} introduced this relaxed notion of recovery under the umbrella of "list decodable robust statistics".       

A first observation is that list decoding is at least as hard as identifying the parameters of mixture models.  With nothing but a planted set of statistical inliers, the outliers can assume any configuration.  A remarkably benign configuration is for the outliers to be arranged as independent mixtures.  In this manner, the gaussian mixture model is a special case of list decodable mean estimation, the mixtures of linear regressions is a special case of list decodable regression, and likewise subspace clustering is a special case of list decodable subspace recovery.  Naturally, any theoretical guarantee for list decodable robust statistics carries directly over to its mixture model counterpart.  Although the converse is evidently false, the chief intellectual thrust of list decodable robust statistics is to establish the settings wherein statistical inference in the presence of overwhelming adversarial outliers is computationally no harder than clustering, a remarkable assertion, especially in light of the settings where list decoding is information theoretically impossible (see eg. \cite{KKK19} \cite{diakonikolas2018list}\cite{kothari2017better}).   

\paragraph{Results}
In this paper we build on a series of works for list decoding of mean estimation \cite{CharikarSV17} \cite{diakonikolas2018list} \cite{kothari2017better}, regression \cite{KKK19} \cite{RY19}, and tackle the natural problem of subspace recovery. Informally, given a dataset for which an $\alpha$ fraction is drawn uniformly in a $k$ dimensional subspace in $d$ dimensions, denoted $U$, and with no additional assumptions on the remaining data, our algorithm outputs a succinct list of $O(\frac{1}{\alpha})$ candidate subspaces one of which is close to the true generative $U$.  Our algorithm is computationaly tractable, runs in polynomial time in both $d$ and $k$.  Furthermore, our algorithm is robust to additive noise, well conditioned linear transformations of the underlying inlier distribution, and succeeds even under the substantially more demanding corruption model where the adversary can simulate any $(1 - \alpha)$ total variation distance corruption of the data.  

Our main algorithmic result is an algorithm for list-decodable subspace recovery.
\begin{theorem} \label{thm:main-alg}
Suppose $\{ X^*_{i} \}_{i \in [N]}$ are drawn from $N(0,I_d)$ and let $P$ be a projection to a $k$-dimensional subspace of $\R^d$.
Let $\{X_i | i \in [N]\}$ be generated by setting 
$$ X_i = P X_i^* + \gamma_i$$
for some additive noise $\gamma_i$ satisfying $$ \sum_{i = 1}^N \norm{\gamma_i}^2 = \epsilon N $$
Let $\{\tilde{X}_i | i \in [N]\}$ be a set of points such that there exists a subset $\cS \in [N]$ of size $|\cS| \geq \alpha N$ with $\tilde{X}_i = X_i$ for all $i \in \cS$.
For all $\eta > 0$, given $N > d^{O(1/\eta^4)}$ samples, there is an algorithm running in time $d^{O(1/\eta^4)}$ that computes a list of $O(1/\alpha)$ projection matrices $\{\Pi_1,\ldots,\Pi_{\ell}\}$ such that 
\[ \min_{j} \norm{P - \Pi_{j}}^2 \leq O\left(\frac{\epsilon}{\eta^2 \alpha^5} + \frac{4ck\eta}{\alpha^5} \right)\]

\end{theorem}
Note that the noise model is perhaps the strongest possible, in that the adversarial corruptions can depend arbitrarily on the samples $\{X_i\}_{i \in [N]}$ and it also includes an additive noise of $\gamma_i$.
For concreteness, if we consider the case with no additive noise ($\epsilon = 0$) then the algorithm recovers a projection $\Pi_j$ with $\norm{P-\Pi_j} \leq O(\sqrt{k})$ in time that is polynomial in $k, d$.   
More generally, the Gaussian distribution $N(0,I_d)$ can be replaced by a distribution whose anti-concentration can be efficiently certified by sum-of-squares proofs (see Appendix for a formal definition).  
Specifically, our results hold for any well conditioned linear transformation of a spherically symmetric distribution with sub-exponential tails (see lemma 9.1 of \cite{RY19}).

Conceptually, we formally state the notion of SoS certifiable resilience, and use it to derive a general algorithm for list-decoding via SoS.  While the ideas behind the general algorithm are implicit in  \cite{KKK19}, we believe the notion of SoS certifiable resilience gives conceptual clarity and might be useful in further applications of the SoS SDPs.  We apply the framework of SoS-certifiable resilience in our presentation of the result for subspace recovery.

Finally we exhibit a lower bound showing that list decodable subspace recovery is impossible even if the inlier distribution is the uniform over the boolean hypercube (see lemma \ref{lem:lowerbound}).

\subsection{Related Work} 

\paragraph{Subspace Recovery} 

Here we discuss related work for the problem of subspace recovery, and highlight key similarities and differences with list decoding. The literature on subspace recovery is vast and we do not attempt a full overview--for a survey, see \cite{EV12}.  Despite the vast literature, the key takeaway is that existing methods for subspace recovery, to the best of our knowledge, fail in the presence of overwhelming adversarial outliers.  

In the worst case setting, \cite{hardt2013algorithms} explore robust subspace recovery in a purely deterministic model where inliers are arranged in general position within a subspace and outliers are in general position in the ambient space.  They provide an algorithm recovering the planted subspace provided $\alpha \geq \frac{k}{d}$.  Their result is essentially optimal as it is Small Set Expansion Hard to recover the subspace if the fraction of outliers is any larger. Although there is no direct comparison with list decodable subspace recovery, the hardness result is solid evidence that subspace recovery in the presence of overwhelming outliers is hard without additional statistical assumptions on the inliers.  Just as worst case assumptions are arguably overly pessimistic, average case assumptions are arguably overly optimistic.  
 
In the statistics literature, subspace clustering is a problem where data is distributed in a union of subspaces in high dimensions where the distribution of points within subspaces, and the relative orientation of subspaces are subject to theoretical assumptions. The goal is to cluster points into their respective subspaces.  This is in contrast with the goals of list decoding, which is a parameter estimation task.    

Statistical approaches model the data according to a mixture of degenerate gaussians. In a sense, this modeling assumption is necessary as subspace clustering is information theoretically impossible even over natural distributional families (see \ref{lem:lowerbound}).  

  A representative approach is Sparse Subspace Clustering (SSC) \cite{EV12} and its robust variant (RSSC) \cite{SEC13} which uses techniques from sparse and low rank recovery algorithms.  RSSC considers subspace clustering in the presence of outliers distributed uniformly on the unit sphere.  This stands in contrast with list decodable subspace recovery where the outliers are adversarially introduced. An overview of spectral clustering algorithms can be found in \cite{V11}.  

Finally, there are subspace clustering algorithms that either lack provable gaurantees or are computationally intractable.  We list a few notable examples.  Generalized Principal Component Analysis \cite{VMS12} \cite{NMCO10} is an algebraic geometric algorithm that treats subspace clustering as a polynomial fitting problem.  Although its recovery guarantees and assumptions are minimal, it's fragile to outliers and its runtime is exponential in $k$.  K-Subspaces \cite{T99} is a generalization of K-means that approaches subspace clustering as a nonconvex optimization.  As a consequence, it is sensitive to initalization and fragile to outliers.  Other iterative algorithms include \cite{AM04}  \cite{BM00}.  Examples of EM style statistical approaches include Mixtures of Probabilistic PCA \cite{TB99}  and other nonconvex approaches include Agglomerative Lossy Compression \cite{MDHW07}.   Unfortunately, it is notoriously difficult to prove the convergence of EM and other nonconvex methods to global optima of the likelihood function. 

\paragraph{List Decodable Learning, Resilience, and the Sum of Squares}
The notion of list decodable learning was introduced by Balcan \etal \cite{BalcanBV08} for clustering problems.  List learning was extended to small $\alpha$ robust statistics in \cite{CharikarSV17}.  They obtained algorithms for list decodable mean estimation, planted partition problems, subsumed under a general stochastic convex optimization framework.    (also see \cite{SteinhardtVC16, SteinhardtKL17}).  
%
%
The same model of {\it list-decodable learning} has been studied for the case of mean estimation \cite{kothari2017outlier} and Gaussian mixture learning \cite{kothari2017better,diakonikolas2018list}.

The notion of {\it resilience} was initialy defined in \cite{steinhardt2017resilience} for robust estimation and extended in \cite{zhu2019generalized}.  
Furthermore, there has been a sequence of works developing the sum of squares method for robust statistics \cite{kothari2017outlier, kothari2017better, kothari2017approximating, hopkins2018mixture}.   

\newcommand{\starE}{\mathop{\mathbb{E}^*}}

\newcommand{\hashE}{\mathop{\mathbb{E}'}}

\newcommand{\Pbool}{P_{\mathsf{bool}}}
\newcommand{\Psum}{P^{(\alpha)}_{\mathsf{sum}}}
\newcommand{\Pcost}{P^{(\epsilon)}_{\mathsf{cost}}}

\section{List Decoding via SoS}

Let $Z_1,\ldots,Z_N$ be samples from a distribution $\cD$ over $\R^d$.
Often, parameters $\Theta^* \in \R^m$ associated with the distribution $\cD$ can be expressed as minima of a cost function associated with each data point.
Specifically, let $\Phi(\Theta,Z)$ be a cost function such that the true parameters of the distribution can be expressed as,
\begin{equation} \label{eq:theta}
    \Theta^* = \argmin_{\Theta \in \cV} \frac{1}{N} \sum_{i \in [N]} \Phi(Z_i,\Theta)
\end{equation}
where $\cV \in \R^m$ is the domain of the parameters.

Since the sum-of-squares proof system can certify facts about low-degree polynomials, we will setup the problem of parameter estimation in this setting.  
First, we assume that $\Phi$ is specified by a polynomial in $\Theta$ and $Z$.  Second, we assume that the parameters $\Theta$ belong to a semi-algebraic set that is specified by a set of polynomial  inequalities, 
$$\cV = \{q_j(\Theta) \geq 0 | 1 \leq j \leq |\cV|\}$$  
Notice that equalities $q(\Theta) = 0$ can be expressed using two inequalities $q(\Theta) \geq 0$ and $-q(\Theta) \geq 0$.
With this setup, the problem of estimating the parameters $\Theta^*$ reduces to solving an optimization problem with polynomial objective and polynomial constraints.

In this work, we will be interested in parameter estimation when an overwhelming fraction of input data is adversarially corrupted.  
Let $\{ \tilde{Z}_1,\ldots,\tilde{Z}_N \}$ be a corrupted data set wherein all but $1-\alpha$-fraction of the samples are adversarially corrupted.  Specifically, for some subset $\cS \subset [N]$ with $|\cS| \geq \alpha N$, we have that
\[ \tilde{Z}_i = \begin{cases}
Z_i & \text{ if } i \in \cS\\
\text{ arbitrary } & \text{ if } i \notin \cS
\end{cases} \]
In the list-decodable recovery problem, we are to recover a small list of candidate assignments $\{\Theta_1,\ldots,\Theta_t\} $ for the parameter such that there exists at least one candidate $\Theta_i$ close to the true value of the parameter $\Theta^*$ on the original data $\{Z_1,\ldots,Z_N\}$.

Parameter estimation in presence of adversarially chosen outliers poses two challenges.  First, the algorithm needs to identify which subset $\cS$ of $\alpha N$ samples $\tilde{Z}_i$ are uncorrupted.  Second, even if the algorithm identifies the subset $\cS$ of samples exactly, it is unclear if the surviving uncorrupted samples still yield a good estimate for the parameters.
The problem of identifying the correct subset of samples $\cS$ can also be posed as a polynomial optimization problem.  The idea is as follows, introduce variables $w_i$ for each sample $\tilde{Z}_i$ to indicate whether the sample is corrupted or not.  Each variable $w_i$ takes a boolean value, which can be enforced by the polynomial constraint 
\[ \text{(Booleanness)} \quad \Pbool(w_i) \defeq w_i^2 - w_i = 0 \]
Furthermore, at least an $\alpha$-fraction of the samples are uncorrupted yielding the constraint
\[ \text{(Sum)} \quad \Psum(\bw) \defeq \sum_{i \in [N]} w_i \geq \alpha N \]
This formulation underlies all applications of SoS SDPs to robust statisics \cite{hopkins2018mixture, kothari2017outlier,kothari2017better}.

Given the subset $\cS \subset \{Z_1,\ldots,Z_N\}$ of uncorrupted samples, a natural estimate of the parameters would be to minimize the total cost for samples within $\cS$.
Specifically, one can construct an estimate $\Theta_{|\cS}$
\begin{equation} \label{eq:thetaS}
    \Theta^*_{|\cS} = \argmin_{\Theta \in \cV} \frac{1}{|\cS|} \cdot \sum_{Z \in \cS} \Phi(Z,\Theta)
\end{equation}
This corresponds to a polynomial constraint of the form,
\[ \text{(Cost)} \qquad \Pcost(\bw,\Theta): \sum_{i\in [N]} w_i \Phi(\tilde{Z}_i, \Theta) \leq \epsilon N \]
Here we use $\epsilon N$ as a generic upper bound, the correct value for the upper bound would depend on the application at hand.

The language of polynomials is very powerful in that a large number of robust parameter estimation problems can be posed as multi-variate constrained polynomial optimization.  On the flipside, it is NP-hard to solve these polynomial optimization problems.  %
The Sum-of-Squares SDP hierarchy and associated sum-of-squares proofs provides a family of efficient algorithms to imperfectly reason about such  systems of polynomials. 

\subsection{Sum-of-Squares SDP hierarchy} \label{SoS-toolkit}

\paragraph{Pseudoexpectations}

%

The sum-of-squares SDP relaxations for a system of polynomial inequalities $\cP$ are a sequence of increasingly stronger SDP relaxations.
The degree $\ell$ SoS SDP relaxation is intended to find the degree $\ell$-moments of a ``probability distribution" over solutions to the system $\cP$.
While the SDP relaxation returns a set of candidate "degree $\ell$ moments" of a distribution, the moments are {\it pseudo-moments} in that there might exist no distribution over solutions to $\cP$ with those moments.
It is notationally convenient to state the degree $\ell$ SoS SDP relaxation in terms of a pseudo-expectation functional $\pE$.

\begin{definition}
A degree $\ell$ pseudoexpectation $\pE : \R[x]_{\leq \ell} \to \R$ {\it satisfying $\cP$} is a linear functional over polynomials of degree at most $\ell$ satisfying 
\begin{enumerate}
\item(Normalization) $\pE[1] = 1$, 
\item(Constraints of $\cP$) $\pE[p(x) a^2(x)] \geq 0$ for all $p \in \cP$ and polynomials $a$ with $\deg(a^2 \cdot p) \le \ell$,
\item(Non-negativity on square polynomials)$\pE[q(x)^2] \ge 0$ whenever $\deg(q^2) \le \ell$.
\end{enumerate}
\end{definition}

For any fixed $D \in \mathbb{N}$, given a polynomial system,one can efficiently compute a degree $D$ pseudo-expectation in polynomial time. 
\begin{fact} (\cite{Nesterov00}, \cite{Parrilo00}, \cite{Lasserre01}, \cite{Shor87}). For any $n$, $D \in \Z^+$, let $\pE_{\zeta}$ be degree $D$ pseudoexpectation satisfying a polynomial system $\cP$.  Then the following set has a $n^{O(D)}$-time weak
separation oracle (in the sense of \cite{GLS1981}):
 \begin{align*}
 & \{\pE_\zeta(1, x_1, x_2, . . . , x_n)^{\otimes D}| \text{ degree } D \text{ pseudoexpectations } \pE_{\zeta} \text{ satisfying }\cP \}
\end{align*}

Armed with a separation oracle, the ellipsoid algorithm finds a degree $D$ pseudoexpectation in time $n^{O(D)}$, which we call the degree $D$ sum-of-squares algorithm. 
\end{fact}

Roughly speaking, the degree $D$-pseudoexpectation functional yields the "degree $D$ moments" of a potential distribution over solutions to the polynomial system.  However, there might not exist any probability distribution with these moments.
Although, the $\pE$ functional does not correspond to an expectation over actual solutions in general, this intuition is useful to keep in mind, and we will appeal to it whenever needed in this overview.
To reason about the properties of pseudo-expectations, one harnesses the dual object namely sum-of-squares proofs.  We turn our attention to sum-of-squares proofs now.
\paragraph{Sum-of-Squares Proofs}
For any nonnegative polynomial $p(x): \R^d \rightarrow \R$, one could hope to prove its nonnegativity by writing $p(x)$ as a sum of squares of polynomials $p(x) = \sum_{i=1}^m q_i(x)^2$ for a collection of polynomials $\{q_i(x)\}_{i=1}^m$.  Such a proof would be succinct and easy to verify.  Unfortunately, there exist nonnegative polynomials with no sum of squares proof even for $d = 2$.  Nevertheless, there is a generous class of nonnegative polynomials that admit a proof of positivity via a proof in the form of a sum of squares.  The key insight of the sum of squares algorithm, is that these sum of squares proofs of nonnegativity can be found efficiently provided the degree of the proof is not too large.  We begin with a rough overview of sum of squares proofs, their dual object pseudoexpectations, and then present the guarantees of the SoS algorithm.      

\begin{definition} (Sum of Squares Proof)
Let $\mathcal{A}$ be a collection of polynomial inequalities $\{p_i(x) \geq 0\}_{i=1}^m$.  A sum of squares proof that a polynomial $q(x) \geq 0$ for any $x$ satisfying the inequalities in $\mathcal{A}$ takes on the form 

\[
    \left(1+ \sum_{k \in [m']} b_k^2(x)\right) \cdot q(x) = \sum_{j\in [m'']} s_j^2(x) + \sum_{i \in [m]} a_i^2(x) \cdot p_i(x) \mcom
\]
where $\{s_j(x)\}_{j \in [m'']},\{a_i(x)\}_{i \in [m]}, \{b_k(x)\}_{i \in [m']}$ are real polynomials.  If such an expression were true, then $q(x) \geq 0$ for any $x$ satisfying $\mathcal{A}$.  We call these identities sum of squares proofs, and the degree of the proof is the largest degree of the involved polynomials $\max \{\deg(s_j^2), \deg(a_i^2 p_i)\}_{i,j}$.  Naturally, one can capture polynomial equalities in $\mathcal{A}$ with pairs of inequalities.   We denote a degree $\ell$ sum of squares proof of the positivity of $q(x)$ from $\cA$ as $\cA \sststile{\ell}{x} \{q(x) \geq 0\}$ where the superscript over the turnstile denote the formal variable over which the proof is conducted.  This is often unambiguous and we drop the superscript unless otherwise specified.    
\end{definition}

Sum of squares proofs can also be strung together and composed according to the following convenient rules.  
\begin{fact}
For polynomial systems $\cA$ and $\cB$, if $\cA \sststile{d}{x} \{p(x) \geq 0\}$ and $\cB \sststile{d'}{x} \{q(x)\geq 0\}$ then $\cA \cup \cB \sststile{\max(d,d')}{x}\{p(x) + q(x) \geq 0\}$.  Also $\cA \cup \cB \sststile{dd'}{x} \{p(x)q(x) \geq 0\}$ 
\end{fact}

Sum of squares proofs yield a framework to reason about the properties of pseudo-expectations, that are returned by the SoS SDP hierarchy.  
%
\begin{fact} (Informal Soundness)
If $\cA \sststile{r}{x} \{q(x) \ge 0\}$ and $\pE$ is a degree-$\ell$ pseudoexpectation operator for the polynomial system defined by $\cA$, then $\pE[q(x)]  \ge 0$.
\end{fact}
\subsection{Certifiable Resilience} \label{sec:framework}

Returning back to our problem of parameter estimation from corrupted data, we need to address the issue that the fragment of uncorrupted data left might be insufficient to faithfully recover the parameter $\Theta$.
More precisely, we will need to make an assumption that the estimate $\Theta^*_{|S}$ (in \eqref{eq:thetaS}) is close to the true estimate $\Theta^*$ (in \eqref{eq:theta})).

The notion of {\it resilience} introduced by  \cite{steinhardt2017resilience} captures this idea.
To exploit the power of sum-of-squares SDP hierarchies, one needs a stronger notion of {\it certifiable resilience}.
A data set is  {\it certifiable resilience} where the dataset is not only resilient, but there is also an efficiently verifiable certificate/proof of its resilience.  In particular, we will be define the notion of {\it Sum-of-Squares certifiable resilience}. 
The formal definition of certifiable resilience is as follows.
\begin{definition} (SoS certifiable resilience)
Fix $\alpha \in (0,1]$ and $\epsilon,\delta > 0$.   
A dataset $\{Z_1,\ldots,Z_N\} \in \R^d$ is said to admit a degree $D$ SoS proof of $(\alpha,\epsilon,\delta)$-resilience if the following polynomial system:
\[
\left\{  \begin{array}{lllr}
\Psum(w) & \defeq    \sum_{i \in [N]}  w_i - \alpha N   & \geq 0 \\
\Pbool(w_i) & \defeq     w_i^2 - w_i & = 0& \forall i \in [N] \\
\Pcost(\bw) & \defeq  \epsilon N -   \sum_{i \in [n]} w_i \Phi(Z_i, \Theta)  & \geq 0 \\
&    q(\Theta) & \geq 0  & \forall q \in \cV
\end{array} \right\} 
\]
can be used to show that $(\sum_i w_i) \cdot \norm{\Theta - \Theta^*}^2 \leq \delta N$ using a sum-of-squares polynomial identity of the form:
\begin{equation} \label{eq:sosCertificate}
 \delta N  - \left(\sum_i w_i\right) \norm{\Theta - \Theta^*}^2 = c(\bw,\Theta) \cdot  \Pcost(\bw,\Theta) + \sum_{i \in [N]} b_i(\bw,\Theta)  \cdot \Pbool(w_i)  + \sum_{q \in \cV} A_q(\bw,\Theta) \cdot q(\Theta) + \lambda \cdot \Psum(\bw) + \lambda_0 \end{equation}
 
where $\lambda,\lambda_0 \in \R^+$, $c(\bw,\Theta)$ and $A_q(\bw,\Theta)$ are sum of squares polynomials and $b_i(\bw,\Theta)$ are arbitrary polynomials in $\bw,\Theta$.  Furthermore, the degree of all the terms in the equality are at most $D$.
\end{definition}
\begin{remark}
We wish to stress on the important distinction between the SoS certificate \eqref{eq:sosCertificate} and the standard notion of SoS proofs.  In \eqref{eq:sosCertificate}, the coefficient of $\Psum$ is necessarily a real number $\lambda$, while a general SoS proof would allow the coefficient of $\Psum$ to be an arbitrary SoS polynomial.  Operationally, if one is constructing the SoS certificate by a proof, this restriction translates to never multiplying $\Psum$ constraint with any polynomial. 
\end{remark}

\subsection{List-decoding}
We will now present an SoS based algorithm for list-decoding the parameters $\Theta$ under the assumption of certifiable resilience.
The essential ingredients of the algorithm are implicit in \cite{KKK19}, but we reformalize the ideas in generality, under the notion of {\it certifiable resilience}.  The notion of certifiable resilience clarifies design of algorithms for list-decodable learning via SoS, and is also useful in presenting our work on subspace recovery.

The general idea behind the algorithm is to solve a sufficiently high-degree SoS SDP relaxation of the polynomial system underlying {\it certifiable resilience}.
This yields a collection of pseudomoments from which we will recover a list of assignments for the parameter $\Theta$, of which one is close to the true value $\Theta^*$.

\paragraph{Frobenius Minimization}
This program faces an immediate bottleneck.
Even if the pseudo-expectation functional corresponded to a true distribution over solutions to the polynomial system, it is conceivable that the distribution does not include the solution $\Theta^*$.  Specifically, the distribution might completely ignore the $\alpha N$ uncorrupted data points, and instead return feasible solutions among the rest.  To overcome this issue, we need to find pseudo-expectations that are {\it comprehensive} in that every valid solution is in their support.
This is achieved by finding a pseudo-expectation functional of maximum entropy or minimum Frobenius norm, among all pseudo-expectation functionals that satisfy the constraints.  This technique first used in the work of Hopkins and Steurer \cite{hopkins2017efficient}, was also used in the two prior works on list-decoding via SoS SDP hierarchy \cite{RY19, KKK19}.
In particular, both these works show that the pseudo-expectation functional that minimizes the Frobenius norm necessarily has good correlation with every possible solution to the polynomial system.  Formally, they show the following.

\begin{lemma}\torestate{ \label{lem:FrobeniusMinimization}
(Comprehensive Pseudodistributions are Correlated with Inliers \cite{RY19,KKK19})
Let $\cP$ be a polynomial system in variables $\bw = \{w_i\}_{i \in [N]}$ and a set of indeterminates $\Theta$, that contains the set of inequalities:
$$\mathcal{P} = \begin{cases} 
\phantom{-} w_i^2 = w_i & \forall i \in [N]\\  

\phantom{-}\sum\limits_{j=1}^{N}w_j  = \alpha N\\

\end{cases}$$
Let $\pE_\zeta: \R[\{w_i\}_{i \in [N]}]^{\leq D} \to \R$ denote a degree $D$ pseudoexpectation that satisfies $\cP$ and minimizes the norm $\norm{\pE_\zeta[ w]}$.
If $\bw' = (w'_1,\ldots,w'_N)  \in \{0,1\}^N$ and $\Theta'$ is a satisfying assignment to $\cP$ then the $\pE_{\zeta}$ has non-negligible support on $\bw'$, 
\begin{equation}
\pE_\zeta\left[\frac{1}{\alpha N} \sum\limits_{i=1}^N w_iw_i' \right] \geq \alpha     
\end{equation}}
\end{lemma}

\paragraph{Rounding}
Assuming we have the moments of a distribution over solutions to the polynomial system, the goal of rounding is to extract each of the solutions to the system.
Both the previous works \cite{RY19,KKK19} employ the idea of conditioning SoS SDP relaxations towards rounding the SDP solution.
Intuitively, the idea is to pick a sample $Z_i$, and condition the pseudoexpectation on the sample $Z_i$ being an inlier, i.e., condition on the event that $w_i = 1$.

Formally, let $\pE : \R\left[ \bw, \Theta  \right] \to \R$ denote the pseudo-expectation functional on polynomials of degree at most $D+1$ in variables $\bw \seteq \{w_i\}_{i \in [N]}$ and $\Theta$.
Pseudo-expectation functionals (equivalently SoS SDP solutions), can be conditioned on low-degree events.  For example, for any $i \in [N]$, the conditioned pseudoexpectation functional $\pE[ \cdot | w_i = 1]$ is constructed as follows,
\begin{equation}
    \pE\left[ p(\bw,\Theta) | w_i = 1\right] \defeq \frac{\pE[p(\bw,\Theta) \cdot w_i]}{\pE[w_i]}  \text{ for all polynomials } p \in \R[\bw,\Theta] \text{ with }\deg(p) \leq D
\end{equation}
               
For a degree $D+1$ pseudoexpectation functional $\pE$, $\pE[ \cdot | w_i = 1]$ is a degree $D$ pseudoexpectation functional.

The two works \cite{RY19,KKK19} analyze the rounding schemes differently, and we follow the simpler analysis in \cite{KKK19}.
We are now ready to formally describe the list-decoding algorithm for certifiably resilient datasets.
\begin{algorithm}[H] \label{algo:listDecoding}
\SetAlgoLined
\KwResult{$\Theta \in \R^m$ such that with probability atleast $\alpha$, $\norm{\Theta - \Theta^*} \leq \delta$}
 \textbf{Inputs}: Parameters $\alpha, \epsilon$ and a corrupted data set $\cD = \{\tilde{Z}_i\}_{i=1}^N$ with $\alpha N$ samples from a $(\alpha,\epsilon,\delta)$-resilient dataset $\{Z_1,\ldots,Z_N\}$, and it admits a degree $D$ SoS certificate of resilience.
 
 Compute the degree $D+1$ pseudo-expectation functional $\pE_{\zeta} : \R[\bw, \Theta] \to \R$ by solving the following SDP. 
\begin{align}
& \underset{\text{degree D pseudoexpectations} \pE}{\text{minimize}}
& & \sum_{i=1}^N \pE[w_i]^2 \\
& \underset{\text{satisfies the polynomial system}}{\text{such that $\pE$}}
&  & (w_i^2 - w_i) = 0, \; i \in [N]\\
& & & \sum_{i=1}^N w_i   \geq \alpha N, \; i \in [N]\\
& & & \sum_{i = 1}^N w_i \Phi(\tilde{Z}_i,\Theta) \leq \epsilon \cdot \left(\sum_{i = 1}^n w_i\right)\\
& & & q(\Theta) \geq 0  \qquad \qquad \forall q \in \cV
\end{align}
\textbf{Rounding:} Sample $i \in [N]$ with probability proportional to $\pE_{\zeta}[w_i]$ and return $\frac{\pE_{\zeta}[w_i \Theta]}{\pE_{\zeta}[w_i]}$
 \caption{list Decoding}
\end{algorithm}

We defer the proof of the following theorem to the appendix.
\begin{theorem} \label{thm:framework-algo}
Fix $\alpha \in (0,1]$ and $\epsilon, \delta > 0$.
Let $Z_1,\ldots,Z_N \in \R^{d}$ be samples that were $(\alpha^2,\epsilon,\delta)$-resilient, and there is a SoS certificate of resilience of degree $D$.  There is an algorithm $\cA$ running in time $O(Nd^{O(D)})$ such that with probability at least $\Omega(\alpha)$, the algorithm outputs $\Theta$ such that $\norm{\Theta - \Theta^*}^2 \leq O(\delta/\alpha^4)$
\end{theorem}

\section{Subspace Recovery}

In this section, we will setup the list-decodable subspace recovery problem and use the framework of certifiable resilience to devise an algorithm for it.
%

Let $\cD$ be a probability distribution over $\R^d$.  For the sake of exposition, we will assume $\cD = N(0, \Id_d)$ but the discussion can be generalized to any SoS-anticoncentrated distribution over $\R^d$.
Let $P$ denote the projection to a $k$-dimensional subspace over $\R^d$.
The uncorrupted data consists of points that are close to projection of $\cD$ to the subspace $P$.  Formally, the uncorrupted data consists of examples $X_i$ of the form, 
$ X_i = P X_i^* + \gamma_i $
where $X_i^*$ is drawn from $\cD$ and $\gamma_i$ is an additive noise.  We will assume that the additive noise is bounded variance in that
$$\sum_{i = 1}^N \norm{\gamma_i}^2 \leq \epsilon N$$

The input to the algorithm consists of $N$ samples $\{ \tilde{X}_i \}_{i \in [N]} \in \R^d$, an $\alpha$-fraction of which are equal to $X_i$.  Specifically, there exists some subset $\cS \in [N]$ such that, $\tilde{X}_i = X_i$ for all $i \in [N]$.  The goal of the list-decoding algorithm is to return a small list of $k$-dimensional subspaces $\{\Pi_1,\ldots,\Pi_\ell\}$ such that at least one of them is close to $P$.

The parameter being estimated here is the $k$-dimensional projection matrix $\Pi$.  
The set of $k$-dimensional projections can be specified by the following set of polynomial equalities in $d \times k$ matrix of . indeterminates $U$
\begin{align*}
    UU^T = \Pi \\
    U^T U = \Id_k 
\end{align*}

A natural estimator for the subspace $P$, if the data were completely uncorrupted would be
\[ \Pi^* = \argmin_{\Pi} \sum_{i \in [N]} \norm{X_i - \Pi X_i}^2 \]
corresponding to the cost function $\Phi(X_i,\Pi) = \norm{X_i - \Pi X_i}^2$.
Thus the polynomial system associated with subspace recovery is given by,
%
%
\begin{equation}\label{eq:polysystemsubspace}
\left\{  \begin{array}{lllr}
\Psum(\bw) & \defeq    \sum_{i \in [N]}  w_i - \alpha N   & \geq 0 \\
\Pbool(w_i) & \defeq     w_i^2 - w_i & = 0& \forall i \in [N] \\
\Pcost(\bw) & \defeq  \epsilon N -   \sum_{i \in [N]} w_i \norm{X_i - \Pi X_i}^2   & \geq 0 \\
\Pi = UU^T \\
U^T U = \Id_k
\end{array} \right\} 
\end{equation}

Through the framework of the previous section, devising an algorithm for subspace recovery (proving \pref{thm:main-alg}) reduces to showing that the data set is certifiably resilient.  
We will sketch the proof of certifiable resilience in this section, and defer the proof of \pref{thm:main-alg} to the Appendix.

\begin{theorem} (SoS certifiable resilience for subspace recovery) \label{thm:subspace-resilience}
For all $\alpha \in (0,1), \eta \in (0,1/2)$ and $\epsilon > 0 $, suppose $\cD$ is a $(c,D(\eta))$-SoS anticoncentrated distribution over $\R^d$ then with high probability, the data set $\{X_1,\ldots,X_N\}$ can be certified to be $(\alpha, \epsilon,\delta)$-resilient by a degree $D(\eta)+4$ SoS certificate for $\delta = \left(\frac{4\epsilon}{\eta^2 \alpha} + \frac{4ck\eta}{\alpha} \right)$.
\end{theorem}

\begin{proof}
For a sample $X_i = P X_i^* + \gamma_i$, we have $X_i - \Pi X_i = (PX_i^* - \Pi P X_i^*)+ (\gamma_i - \Pi \gamma_i)$.  We can rewrite it as,
\[ PX_i^* - \Pi P X_i^* = (X_i - \Pi X_i) - (\gamma_i - \Pi \gamma_i) \]
Using $(a+b)^2 \leq 2 a^2 + 2 b^2$ we get that,
\begin{align*}
    \norm{P X^*_i - \Pi P X^*_i}^2 & = 2 \norm{X_i - \Pi X_i}^2 + 2 \norm{\gamma_i -\Pi \gamma_i}^2 \leq 2 \norm{X_i -\Pi X_i}^2 + 2 \norm{\gamma_i}^2
\end{align*}
where the last inequality uses the fact that $\Pi^2 = \Pi$.  Hence we get that,
\begin{align} \label{eq:myeq23}
    \sum_{i \in [N]} w_i \norm{P X^*_i - \Pi P X^*_i}^2 & \leq 2 \sum_{i \in [N]} w_i \norm{X_i - \Pi X_i}^2 + w_i \norm{\gamma_i}^2  \leq 2  \epsilon N + 2\epsilon N.
\end{align}
where the second inequality uses $\Pcost(\bw) \geq 0$, $w_i \leq 1$ for all $i$ and $\sum_i \norm{\gamma_i}^2 \leq \epsilon N$.
Fix an orthonormal basis $e_1,\ldots, e_d$ such that $\Span\{e_1,\ldots,e_k\} = P$.
\begin{align*}
 \norm{(P - \Pi P) X^*_i }^2
 = \sum_{\ell = 1}^d \langle  (P - P \Pi )e_\ell , X^*_i \rangle^2 \geq \sum_{j = 1}^k \langle  (P - P \Pi )e_\ell , X^*_i \rangle^2 
\end{align*}
Using \eqref{eq:myeq23} we get that,
\begin{equation} \label{eq:myeq1}
 \sum_{i \in [N]} w_i    \sum_{\ell = 1}^k \langle  (P - P \Pi )e_\ell , X^*_i \rangle^2 \leq 4\epsilon N
\end{equation}

Suppose $v_j \seteq (P-P\Pi)e_j$, then its norm $\norm{v_j}^2 = \norm{(P - P\Pi)e_j}^2 \leq 2 \norm{Pe_j}^2 + 2 \norm{ P\Pi e_j}^2 \leq 4$.  Since $X_i^*$ is drawn from a $(c,D(\eta))$-anticoncentrated distribution, there is a degree $D(\eta)$ SoS derivation for 
\begin{align*}
\sum_{j \in [N]} w_i \iprod{v_j, X_i^*}^2 \geq \eta^2 \left(\sum_{i \in [N]} w_i\right) \norm{v_j}^2 - 4 c \eta^3 N
\end{align*}
Summing the above inequality over all $j = 1 \ldots k$ we get that,

\begin{align*}
\sum_{\ell = 1}^k \sum_{i \in [N]} w_i \iprod{(P-P\Pi)e_\ell, X_i^*}^2 \geq \eta^2 \left(\sum_{i \in [N]} w_i\right) \left(\sum_{\ell} \norm{(P-P\Pi)e_\ell}^2\right) - 4 c k\eta^3 N
\end{align*}

In conjunction with \eqref{eq:myeq1}, this implies that,
\begin{align*}
 \left(\sum_{i \in [N]} w_i\right) \left(\sum_{\ell} \norm{(P-P\Pi)e_\ell}^2\right) \leq \frac{1}{\eta^2} \sum_{\ell = 1}^k \sum_{i \in [N]} w_i \iprod{(P-P\Pi)e_\ell, X_i^*}^2 + 4 c k\eta N  \leq \left(\frac{4\epsilon}{\eta^2} + 4ck \eta\right)N
\end{align*}
Note that $\sum_{\ell = 1}^k \norm{(P-P\Pi)e_\ell}^2 = \norm{P-P\Pi P}_F^2$.  Using the constraint $\Psum(\bw)$, we can rewrite the above equation as,
\begin{align*}
 \left(\sum_{i \in [N]} w_i\right)\norm{P-P\Pi P}_F^2  \leq \left(\frac{4\epsilon}{\eta^2 \alpha} + \frac{4ck \eta}{\alpha}\right) \left(\sum_{i \in [N]} w_i\right) 
 \end{align*}
By \pref{lem:mylemma1} (see Appendix for proof) this implies that,
\begin{align*}
 \left(\sum_{i \in [N]} w_i\right)\norm{P-\Pi}_F^2  \leq \left(\frac{4\epsilon}{\eta^2 \alpha} + \frac{4ck \eta}{\alpha}\right) \left(\sum_{i \in [N]} w_i\right) 
 \end{align*}
\end{proof}

\begin{lemma} \label{lem:mylemma1}
For a $k$-dimensional projection matrix $P \in \R^{d \times d}$ and a $d \times k $ matrix of indeterminates $U$ and $\gamma > 0$,
\[ 
\left\{\begin{array}{l}
     (\sum_i w^2_i) \cdot \norm{P - P\Pi P}_{F}^2 \leq k \gamma (\sum_{i} w^2_i) \\
     \Pi= UU^T \\
     U^T U = \Id_k \\
     \end{array}\right\} \SoSp{4} 
\left(\sum_i w^2_i\right) \cdot \norm{P-\Pi}_F^2 \leq k \gamma \left(\sum_i w^2_i\right)
\]
\end{lemma}

The algorithm produced by the framework in previous section will output a matrix $\Pi_i = \frac{\pE[w_i \Pi]}{\pE[w_i]}$.  The output matrix $\Pi_i$ satisfies $0 \preceq \Pi_i \preceq I$ and $\Tr[\Pi_i] = k$, but is not necessarily a projection matrix.  In order to recover a projection matrix, we will have to round the matrix further into one.  The following lemma shows that just picking the projection onto the top $k$ eigenvalues of $\Pi_i$ yields a projection matrix with only a constant factor loss in the error.

\begin{lemma} \label{lem:eigenspace}
Let $\Pi$ be a matrix satisfying $\Tr (\Pi) = k$ and $\Pi \preceq I$.  Let $P$ be a rank $k$ projection.  If $\langle \Pi, P \rangle \ge k (1- \varepsilon)$, then 
$\langle \Pi_k, P \rangle \ge k (1- 2\varepsilon)$ where $\Pi_k$ is the top $k$ eigenspace of $\Pi$. Equivalently $\norm{P - \Pi_k}_F^2 \leq 4\epsilon$
\end{lemma}

\begin{proof} 
Let $\lambda_1,...,\lambda_d$ and $v_1,...,v_d$ be the eigenvalues and eigenvectors of $\Pi$.  Then $\langle \Pi, P \rangle \ge k (1- \varepsilon)$ imlies  $\sum_{j=1}^k \lambda_j \ge k (1- \varepsilon)  \text{ and therefore} \sum_{j=k+1}^d \lambda_j \le  \varepsilon k$.  Thus we have the following series of inequalities.  
\begin{align*} 
k (1- \varepsilon) & \le \langle \Pi, P \rangle  =  \sum_{j=1}^d \lambda_j \langle  P, v_j v_j^T \rangle =
\sum_{j=1}^k \lambda_j \langle  P, v_j v_j^T \rangle + \sum_{j=k+1}^d \lambda_j \langle  P, v_j v_j^T \rangle 
\end{align*}

\begin{align*}
\leq \sum_{j=1}^k \lambda_j \langle  P, v_j v_j^T \rangle + \sum_{j=k+1}^d \lambda_j  \le \sum_{j=1}^k \lambda_j \langle  P, v_j v_j^T \rangle +  \varepsilon k
\le \sum_{j=1}^k  \langle  P, v_j v_j^T \rangle +  \varepsilon k= \langle \Pi_k, P \rangle +  \varepsilon k.
\end{align*} 
Rearranging, we obtain, $\langle \Pi_k, P \rangle \ge k (1- 2\varepsilon)$. 

\end{proof}

\newcommand{\ur}{\underline{r}}

\newcommand{\Gc}{\tilde G}
\newcommand{\Ac}{\tilde A}
\newcommand{\Hcmu}{\tilde H^\mu}
\newcommand{\Acmu}{\tilde A^\mu}
\newcommand{\Gcmu}{\tilde G^\mu}
\newcommand{\Xcmu}{\tilde X^\mu}
\newcommand{\Hmu}{H}
\newcommand{\Amu}{A}
\newcommand{\Gmu}{G}
\newcommand{\Xmu}{X}
\newcommand{\Dmumuc}{D}

\newcommand{\Gammamu}{\Gamma_\mu}
\newcommand{\Gammamup}{\Gamma_{\mu+1}}

\newcommand{\goodev}{\Xi}
\newcommand{\badev}{\Xi^c}
\newcommand{\indgoodmu}{\phi_\mu}
\newcommand{\indbadmu}{\bar\phi_\mu}
\newcommand{\indgood}{\phi}
\newcommand{\indbad}{\bar\phi}
\newcommand{\chimu}{\chi_\mu}

\newcommand{\hfe}{{\bf \hat e}}

\newpage

\section{Hardness of List Decodable Subspace Recovery over Hypercube}

We construct a dataset for which list decodable subspace recovery is impossible.  The argument is straightforward and follows from a Gilbert-Varshamov style lower bound.   Our inliers will be distributed in a $k$ dimensional subspace according to the uniform distribution over the boolean hypercube. This innocuous setup turns out to be impossible to list decode even when the outliers are arranged in a benign mixture model distributed uniformly over the boolean hypercube in $\frac{1}{\alpha}$ orthogonal subspaces.  For simplicity of presentation, we will assume that each corner of the hypercube is populated by the same number of points.  A key takeaway of our lower bound is that modeling inliers as a standard normal over a planted subspace, in addition to being a popular statistical choice, is in a sense necessary.  Distributions that form distinct clumps of points are subject to pathologies of interpolation, a problem that is mitigated for points that are anticoncentrated.         

\begin{lemma} \label{lem:lowerbound}
Let $\alpha \in [0,1/2]$ be a fixed constant.  For any $d \geq \frac{k}{\alpha}$, let $r_1,...,r _N \sim \{0,1\}^{\frac{k}{\alpha}}$ be a set of points such that there are an equal number on each corner of the $\frac{k}{\alpha}$ dimensional hypercube.  Then let the dataset $X_1,...,X_N \in \R^d$ be defined such that $X_i := (r_i, 0^{d - k/\alpha})$ for all $i \in [N]$ where each datapoint $X_i$ is formed by padding the end of its corresponding $r_i$ vector with zeros. Then there exists a list $L = \{P_i\}_{i=1}^q$ of projection matrices $P_i$ onto $k$ dimensional subspaces where each $P_i \in L$ contains at least $\alpha N$ points; any pair of projections $P_i,P_j \in L$, satisfies $\norm{P_i - P_j}_F \geq \sqrt{2\epsilon k}$ for a constant $\epsilon \in [0,1/2]$; the length of $|L| = q$ is greater than $(\frac{1}{\alpha2^{H(\epsilon)}})^k$ where $H(\epsilon)$ is the binary entropy function.  Thus for $\epsilon < \frac{1}{2}$ there exists a $k$ such that no list decoding algorithm can succeed with a fixed polynomial list length.      
\end{lemma}

\begin{proof}
Let $V := \{e_1,...,e_{k/\alpha}\}$ be the first $k/\alpha$ basis vectors.  There exists $\frac{k}{\alpha} \choose k$ subspaces comprised of $k$ basis vectors from $V$.  We can encode these subspaces as vectors over $\{0,1\}^{\frac{k}{\alpha}}$ where a $1$ indicates the presence of a basis vector, and a zero indicates the absence.  Proving the theorem then reduces to proving that the maximum size of binary code $C$ of length $\frac{k}{\alpha}$ of hamming weight exactly $k$ with decoding radius $\epsilon k$ is lower bounded by $(\frac{1}{\alpha2^{H(\epsilon)}})^k$.  

The proof follows from classical arguments.  Since $C$ is of maximum size, there does not exist a boolean vector $c_x$ of hamming weight $k$ that is further than $\epsilon k$ away from its closest codeword $c \in C$.  Otherwise, it would be possible to add $c_x$ into $C$ which is a contradiction.  Since there are exactly  $\frac{k}{\alpha} \choose k$ vectors of hamming weight exactly $k$, and the hamming ball of radius $\epsilon k$ has exactly $\sum_{i=0}^{\epsilon k} {k \choose i}$ the size of $C$ is lower bounded by 

$$ |C| \geq \frac{{\frac{k}{\alpha} \choose k}}{\sum_{i=0}^{\epsilon k} {k \choose i}}$$
we lower bound the numerator by $(\frac{1}{\alpha})^k$, and upper bound the denominator by the binomial theorem $\sum_{i=0}^{\epsilon k} {k \choose i} \leq 2^{H(\epsilon)k}$.  Thus we have shown $|C| \geq (\frac{1}{\alpha2^{H(\epsilon)}})^k$ as desired.  
\end{proof}

\newpage
\bibliographystyle{amsalpha}
\bibliography{mr,dblp,scholar,bibliography,listDecoding}

\newpage
\appendix

\section{SoS-Certifiable Anticoncentration}

Here we recall the notion of SoS-certifiable anti-concentration and the relevant results from \cite{RY19}.

\begin{definition}\label{def:anticoncentration} 
Let $D: [0,1/2] \to \N$.
A probability distribution $\cD$ over $\R^d$ is said to $(c, D(\eta))$-SoS-anticoncentrated, 
If for any $0 < \eta < \frac{1}{2}$ there exists $\tau \leq c\eta$ and there exists a constant $k \in \N$ such that for all $N > d^k$,
with probability $1-d^{-k}$, over samples $x_1,\ldots, x_N \sim \cD$ the following polynomial system

$$
\cP = \left\{
        \begin{array}{ll}
            w_i^2 = w_i & i \in [N]\\
    \norm{v}^2 \leq \rho^2 \\
        \end{array}
    \right.
$$
yields a degree $D(\eta)$ SoS proof of the following inequality 

\begin{align*}
    \cP\SoSp{D(\eta)}\Big\{\frac{1}{N}\sum_{i = 1}^N w_i \iprod{X_i,v}^2 \geq \eta^2 (\frac{1}{N}\sum_i w_i) \norm{v}^2 - \eta^2 \tau \rho^2\Big\}
\end{align*}
\end{definition}

\begin{theorem}\label{thm:anti-sufficient} (Sufficient conditions for SoS anti-concentration)
If the degree $D(\eta)$ empirical moments of $\cD$ converge to the corresponding true moments $M_t$ of $\cD$, that is for all $t \leq D(\eta)$ 
$$
\Norm{\frac{1}{N}\sum_{i=1}^N X_i^{\otimes \frac{t}{2}}(X_i^{\otimes \frac{t}{2}})^T - M_t}_F \leq d^{-k} 
$$
And if there exists a uni-variate polynomial $I_{\eta}(z) \in \R[z]$ of degree at most $D(\eta)$ such that 
\begin{enumerate}
    \item $I_{\eta}(z) \geq 1-\frac{z^2}{\eta^2\rho^2}$ for all $z \in \R$.
    \item $\cP\SoSp{D(\eta)} \Big\{\norm{v}^2 \cdot \E_{x \in \cD}[I_{\eta}(\iprod{v, x})] \leq c\eta\rho^2\Big\}$.
\end{enumerate}
Then $\cD$ is $(c,D(\eta))$ certifiably anticoncentrated.

\end{theorem}

\begin{lemma}\torestate{\label{lem:normal-anti}
For every $d \in \N$, the standard Gaussian distribution $\cN(0,I_d)$ is $(c,O(\frac{1}{\eta^4}))$-SoS-anticoncentrated. In particular there exists a construction for $c \leq 2\sqrt{e}$}    
\end{lemma}

\section{Sum-of-Squares Toolkit}
Here we present some useful inequalities captured by the sum of squares proof system

\paragraph{Useful Inequalities}

\begin{fact} (Cauchy Schwarz)
Let $x_1,..,x_n,y_1,...,y_n$ be indeterminates, than 

$$\sststile{4}{} \big(\sum_{i \leq n}x_iy_i\big)^2 \leq \big(\sum_{i \leq n}x_i^2\big)\big(\sum_{i \leq n}y_i^2\big) $$
\end{fact}

\begin{fact} (Triangle Inequality) 
Let $x,y$ be $n$-length vectors of indeterminates, then 
$$\sststile{2}{} \norm{x + y}^2 \leq 2\norm{x}^2 + 2\norm{y}^2 $$
\end{fact}

\begin{fact}\label{fact:pseudocauchy}(Pseudoexpectation Cauchy Schwarz). 
Let $f(x)$ and $g(x)$ be degree at most $\ell \leq \frac{D}{2}$ polynomial in indeterminate $x$, then $$\pE[f(x)g(x)]^2 \leq \pE[f(x)^2]\pE[g(x)^2]$$  
\end{fact}

\begin{fact} (Spectral Bounds)
Let $A \in \R^{d \times d}$ be a positive semidefinite matrix with $\lambda_{max}$ and $\lambda_{min}$ being the  largest and smallest eigenvalues of $A$ respectively. Let $\pE$ be a pseudoexpectation with degree greater than or equal to $2$ over indeterminates $v = (v_1,...,v_d)$.  Then we have 
$$\sststile{2}{} \iprod{A,vv^T} \leq \lambda_{max} \norm{v}^2 $$
and 
$$\sststile{2}{} \iprod{A,vv^T} \geq \lambda_{min} \norm{v}^2 $$
\end{fact}

\section{Omitted Proofs}

\subsection{Proof of \pref{thm:main-alg}}
\begin{theorem}
Suppose $\{ X^*_{i} \}_{i \in [N]}$ are drawn from $N(0,\Id_d)$ and let $P$ be a projection to a $k$-dimensional subspace of $\R^d$.
Let $\{X_i | i \in [N]\}$ be generated by setting 
$$ X_i = P X_i^* + \gamma_i$$
for some additive noise $\gamma_i$ satisfying $$ \sum_{i = 1}^N \norm{\gamma_i}^2 = \epsilon N $$
Let $\{\tilde{X}_i | i \in [N]\}$ be a set of points such that there exists a subset $\cS \in [N]$ of size $|\cS| \geq \alpha N$ with $\tilde{X}_i = X_i$ for all $i \in \cS$.
For all $\eta > 0$, there is an algorithm running in time $d^{O(1/\eta^4)}$ that computes a list of $O(1/\alpha)$ projection matrices $\{\Pi_1,\ldots,\Pi_{\ell}\}$ such that 
\[ \min_{j} \norm{P - \Pi_{j}}^2 \leq O\left(\frac{\epsilon}{\eta^2 \alpha^5} + \frac{4ck\eta}{\alpha^5} \right)\]

\end{theorem}
\begin{proof}
The theorem is a consequence of applying the framework from \pref{sec:framework} to the polynomial system \eqref{eq:polysystemsubspace}.
Specifically, in \pref{thm:subspace-resilience} we show that for any $(c,D(\eta))$-SoS anticoncentrated distribution, the data set $\{X_1,\ldots,X_N\}$ is SoS-certifiably $(\alpha,\epsilon,\delta)$-resilient for $\delta = \left(\frac{4\epsilon}{\eta^2 \alpha} + \frac{4ck\eta}{\alpha} \right)$ and degree $D(\eta)+ 4$.
Using the algorithm in \pref{thm:framework-algo}, we can recover a matrix $\Pi$ such that $\norm{P-\Pi}_F^2 = \delta/\alpha^4$.
While $\Pi$ satisfies $\Pi \preceq \Id$ and $\Tr(\Pi) = k$, it is not necessarily a projection matrix.  In particular, $\Pi$ can have eigenvalues that are not $0$ or $1$.
However, we prove in \pref{lem:eigenspace} the matrix $\Pi$ can be rounded to a projection matrix with only a constant loss in the squared Frobenius norm $\norm{P-\Pi}_F^2$.
Thus one can recover a subspace $\Pi$ such that $\norm{P-\Pi}_F^2 \leq O\left( \frac{\epsilon}{\eta^2 \alpha^5} + \frac{4ck\eta}{\alpha^5} \right)$.
The running time of the algorithm is $d^{O(D(\eta))}$.  By \pref{lem:normal-anti}, the Gaussian distribution is $(2 \sqrt{e}, O\left(\frac{1}{\eta^4}\right))$-SoS anticoncentrated, thus giving a runtime of $d^{O(1/\eta^4)}$.

\end{proof}

\subsection{Proof of \pref{thm:framework-algo}}

\begin{proof}
For sake of succinctness, we will denote 
\[ \pE_i \defeq \pE[ \cdot | w_i = 1] \qquad 
\text{ and }
 \beta_i \defeq \pE[w_i] \mper \]


Define the pseudoexpectation operator $\starE$ by modifying $\pE$ to make $w_j = 0$ for all $j \notin \cS$.  In particular, for any $j \notin \cS$, pseudoexpectation of all monomials containing $w_j$ is set to $0$.  Formally, for every monomial we set,
It is easy to check that $\starE[w_i^2] = \starE[w_i]$, and that $\starE$ satisfies $\starE[q(\Theta)] = 0$ for all $q \in \cV$.  Finally, the cost of the $\starE$ on true data $\{Z_1,\ldots,Z_N\}$ is not higher than the cost of $\pE$.  This is because,
\begin{align}\label{eq:starpcost} \starE\left[ \sum_{j \in [N]}  w_j \Phi(Z_j,\Theta)\right]  = \sum_{j \in \cS} \pE[w_j \Phi(Z_j, \Theta)]  
\leq \sum_{j \in [N]} \pE[w_j \Phi(\tilde{Z}_j, \Theta)]  \leq \epsilon N \end{align}
where we used the fact that $\pE[w_j \Phi(\tilde{Z}_j,\Theta)] \geq 0$ for all $j \in [N]$ and the $\Pcost$ constraint on the corrupted data,

Notice that the rounding algorithm outputs $\pE_i[\Theta]$ with probability proportional to $\beta_i = \pE[w_i]$.  
By \pref{lem:FrobeniusMinimization}, we have that  that,
\[ \sum_{i \in \cS} \beta_i \geq \pE_{\zeta}[\sum_{i \in \cS} w_i] \geq \alpha^2 N\]
while $\sum_{i \in [N]} \beta_i \leq N$. %
Therefore with probability at least $\alpha^2$, the rounding algorithm picks $i \in \cS$.  
Conditioned on picking an element $i \in \cS$, the expected distance $\norm{\Theta - \Theta^*}^2$ satisfies,
\[
\begin{array}{lr}
\left(\E_{i \sim \cS}[ \norm{\pE_i[\Theta] - \Theta^*}]\right)^2 & \\
 = \left(\E_{i \sim \cS}[ \norm{\starE_i[\Theta] - \Theta^*}]\right)^2  & \text{ Definition of } \starE \\
  = \left( \frac{1}{\sum_{i \in \cS} \beta_i} \sum_{i \in \cS} \beta_i \norm{\starE_i[\Theta -\Theta^*]} \right)^2 & \text{Definition of } \E_{i \in \cS} \\
  = \left( \frac{1}{\sum_{i \in \cS} \beta_i} \right)^2 \cdot \left(\sum_{i \in \cS} \norm{\starE[w_i (\Theta - \Theta^*)]} \right)^2 & \text{ using } \pE_i [ \Theta] = \frac{\pE[w_i \Theta]}{\pE[w_i]}\\
  \leq \left( \frac{1}{\sum_{i \in \cS} \beta_i} \right)^2 \cdot |\cS| \cdot \sum_{i \in \cS} \norm{\starE[w_i (\Theta - \Theta^*)]}^2 & \text{ using Cauchy-Schwartz} \\
 \leq \left( \frac{1}{\sum_{i \in \cS} \beta_i} \right)^2 \cdot |\cS| \cdot \sum_{i \in \cS} \starE \left[\norm{w_i (\Theta - \Theta^*)}^2\right] & \text{ using pseudo-expectation Cauchy-Schwartz} \\
 = \left( \frac{1}{\sum_{i \in \cS} \beta_i} \right)^2 \cdot | \cS | \cdot \starE\left[ \left(\sum_{i \in \cS} w_i\right) \norm{\Theta - \Theta^*}^2 \right] & \text{using } w_i^2 - w_i = 0\\
 = \left( \frac{1}{\sum_{i \in \cS} \beta_i} \right)^2 \cdot | \cS | \cdot \starE\left[ \left(\sum_{i \in [N]} w_i\right) \norm{\Theta - \Theta^*}^2 \right] & \text{using } \starE[w_i] = 0 \text{ for } i \notin \cS\\
 \leq \frac{\delta}{\alpha^4}
\end{array}
\]

Finally, note that $\starE$ is a valid degree $D$ pseudo-expectation functional that satisfies the constraints $\Pbool(w_i)$ and $\Pcost$ (\eqref{eq:starpcost}) on the original data $\{Z_1,\ldots,Z_N\}$.
Since the uncorrupted data $\{Z_1,\ldots,Z_N\}$ is $(\alpha,\epsilon,\delta)$-resilient and admits a degree $D$ SoS certificate, we can conclude that
\[ \starE\left[ \left(\sum_{i \in [N]} w_i\right) \norm{\Theta - \Theta^*}^2 \right] \leq \delta N\]
Along with the fact that $\sum_i \beta_i \geq \alpha^2 N$, the above calculation implies that 
\begin{align}
\left(\E_{i \sim \cS}[ \norm{\pE_i[\Theta] - \Theta^*}]\right)^2 \leq \frac{\delta}{\alpha^4} 
\end{align}
\end{proof}

\subsection{Proof of \pref{lem:mylemma1}}

\begin{proof}
For notational convenience, denote $A \defeq \sum_i w_i^2$.  Note that
\begin{align*}
2 P \Pi P = P^2 + (P\Pi P)^2 - (P-P\Pi P)^2 \mper
\end{align*}
Thus 
\begin{align*}
 2 A \cdot \Tr[P \Pi P] & = A  \Tr[P^2] + A\Tr[(P\Pi P)^2] - A  \Tr[(P-P\Pi P)^2]\\
& = A k + A k - A \Tr[(P-P\Pi P)^2] \qquad (\text{\pref{lem:ppip}}) \\
& = A k + A k - A k\gamma \qquad (\text{\pref{lem:ppip}}) \\
& = k (2-\gamma) \cdot A
\end{align*}
Finally, we have
\begin{align*}
A \cdot \norm{P-\Pi}_F^2 = A (\Tr[P^2] + \Tr[\Pi^2] - 2 \Tr(P \Pi)) \\
= A \left(2k - 2 \Tr(P\Pi P)\right)
\leq k \gamma \cdot A
\end{align*}

\end{proof}

\begin{lemma} \label{lem:ppip}
Suppose $U$ is a $d \times k$ matrix of indeterminates satisfying $U^T U = \Id_k$.  Suppose $\Pi = UU^T$ and $P \in \R^{d \times d}$ is a projection matrix then, 

\begin{align}
    \begin{array}{c}
         U^T U = \Id_k \\
         \end{array} \SoSp{2} \Tr[(P\Pi P)^2] \leq k
\end{align}
\end{lemma}
\begin{proof}

Observe that for any positive semidefinite matrices $A,B \succeq 0$ and indeterminates $U$
\begin{align*}
    \Tr[U^T A U U^T B U] & = \Tr[ U^T A^{1/2} A^{1/2} U U^T B^{1/2} B^{1/2} U] \\
    & = \Tr[(A^{1/2} UU^T B^{1/2}) (A^{1/2} UU^T B^{1/2})^T ] \\
    & = \norm{A^{1/2} UU^T B^{1/2}}_F^2 \geq 0
\end{align*}
Now we can write,
\begin{align*}
    \Tr[(P\Pi P)^2] &= \Tr[P UU^TPUU^TP]  \\
    & = \Tr[U^T P U U^T P U] \\
    & \leq  \Tr[U^{T} P U U^T P U] + \Tr[U^{T} P U U^T (\Id-P) U] + \Tr[U^{T} (\Id-P) U U^T  U]  \\
    & = \Tr[(U^{T} U)^2] \\
    & = \Tr[\Id_k] = k  \qquad \text{using } U^TU = \Id_k
\end{align*}
\end{proof}


\end{document}